%% file: sample-sigconf.tex
\documentclass[sigconf]{acmart}
\usepackage{enumitem}
\usepackage{hyperref}
\usepackage{multirow}
\usepackage{amsthm}
\usepackage{array}
\usepackage{subfigure}
\usepackage{graphicx}
\setlist[enumerate]{align=left}
\usepackage{algpseudocode}
\usepackage{algorithm}
\usepackage{balance}
\usepackage{amsthm}
\newtheorem{theorem}{Theorem}
\newtheorem{assumption}{Assumption}
\usepackage{subcaption}
\usepackage[multiple]{footmisc}

\AtBeginDocument{%
  }

\copyrightyear{2026}
\acmYear{2026}
\setcopyright{cc}
\setcctype{by}
\acmConference[SIGIR '26]{Proceedings of the 49th International ACM SIGIR Conference on Research and Development in Information Retrieval}{July 20--24, 2026}{Melbourne, VIC, Australia}
\acmBooktitle{Proceedings of the 49th International ACM SIGIR Conference on Research and Development in Information Retrieval (SIGIR '26), July 20--24, 2026, Melbourne, VIC, Australia}
\acmDOI{10.1145/3805712.3809654}
\acmISBN{979-8-4007-2599-9/2026/07}

\begin{document}

\title{FedMM: Federated Collaborative Signal Quantization for Multi-Market CTR Prediction}

\author{Jun Zhang}
\authornote{Both authors contributed equally to this research.}
\affiliation{%
  \institution{College of Computer Science and Software Engineering, Shenzhen University}
  \institution{Shenzhen Technology University}
  \city{Shenzhen}
  \country{China}
}
\email{2410104014@mails.szu.edu.cn}

\author{Dugang Liu}
\authornote{Corresponding authors.}
\affiliation{%
  \institution{College of Computer Science and Software Engineering, Shenzhen University}
  \city{Shenzhen}
  \country{China}}  
\email{dugang.ldg@gmail.com}

\author{Xing Tang}
\authornotemark[1]
\affiliation{%
  \institution{Shenzhen Technology University}
  \city{Shenzhen}
  \country{China}}  
\email{xing.tang@hotmail.com}

\author{Xiuqiang He}
\authornotemark[2]
\affiliation{%
  \institution{Shenzhen Technology University}
  \city{Shenzhen}
  \country{China}}  
\email{hexiuqiang@sztu.edu.cn}

\author{Zhong Ming}
\affiliation{%
  \institution{Shenzhen University}
  \city{Shenzhen}
  \country{China}}  
\email{mingz@szu.edu.cn}
\renewcommand{\shortauthors}{Jun Zhang et al.}

\input{section/abstract}

\ccsdesc[500]{Information systems~Recommender systems}
\ccsdesc[500]{Security and privacy~Privacy protections}

\keywords{Signal quantization, Click-through rate prediction, Multi-market, Privacy-preserving}

\maketitle

\input{section/introduction}
\input{section/relatedwork}
\input{section/formulation}
\input{section/theoretica}
\input{section/method}
\input{section/experiments}

\input{section/conclusion}
\input{section/acknowledgement}

\bibliographystyle{ACM-Reference-Format}
\balance
\bibliography{sample-base}

\end{document}

%% file: section/abstract.tex
\begin{abstract}
Online platforms such as Amazon and Netflix serve users across multiple countries and regions, underscoring the importance of multi-market recommendation (MMR). Most MMR methods adopt a pre-training and fine-tuning paradigm, in which a unified model is first trained on centralized, global data and subsequently adapted to specific markets. However, this approach ignores the privacy of market data. While traditional federated learning preserves privacy, it typically aims to obtain a global model by aggregating model parameters and does not account for significant market heterogeneity. Additionally, because ID spaces are disjoint across markets, embedding-based aggregation strategies become ineffective. 
To overcome these challenges, we propose a federated collaborative signal quantization (FedMM) method for multi-market click-through rate (CTR) prediction. Our core idea leverages a discrete codebook mechanism to achieve privacy-preserving transmission and align disjoint ID spaces. We further employ a hierarchical codebook structure to capture cross-market shared patterns and market-specific characteristics. Specifically, we deploy a residual quantized variational autoencoder (RQ-VAE) with a dual-layer codebook mechanism for each market to quantize collaborative embeddings. The first layer utilizes a global federated codebook, updated via aggregation to capture universally shared collaborative patterns, while the second layer maintains a local codebook to learn market-specific semantics. Finally, the learned discrete codes, which integrate both general and specific collaborative signals, are incorporated into downstream CTR models to enhance prediction accuracy across all markets. Extensive experiments on benchmark datasets demonstrate that FedMM significantly improves recommendation performance with privacy guarantees.
\end{abstract}

%% file: section/introduction.tex
\section{Introduction}\label{sec:intro}
With the rapid globalization of the digital economy, online platforms such as Amazon~\cite{amazon}, Netflix~\cite{netflix}, and TikTok are expanding their services across borders. As shown in Figure~\ref{fig:intro}, these platforms serve diverse user groups across countries and regions (referred to as markets), with users exhibiting distinct cultural backgrounds and consumption preferences. To meet the personalized needs of users across different markets and improve user experience and platform revenue, it is essential to build a recommendation system that adapts to multiple markets. Consequently, Multi-market recommendation (MMR) has emerged as a key research direction~\cite{xmarket}. MMR deploys a unified framework to serve diverse user groups distributed across geographically distinct regions. By modeling the complex dependencies and heterogeneity across markets, MMR aims to deliver accurate and personalized content to users in each specific market.

\begin{figure}[htbp] 
\centering
    \includegraphics[width=1.0\linewidth]{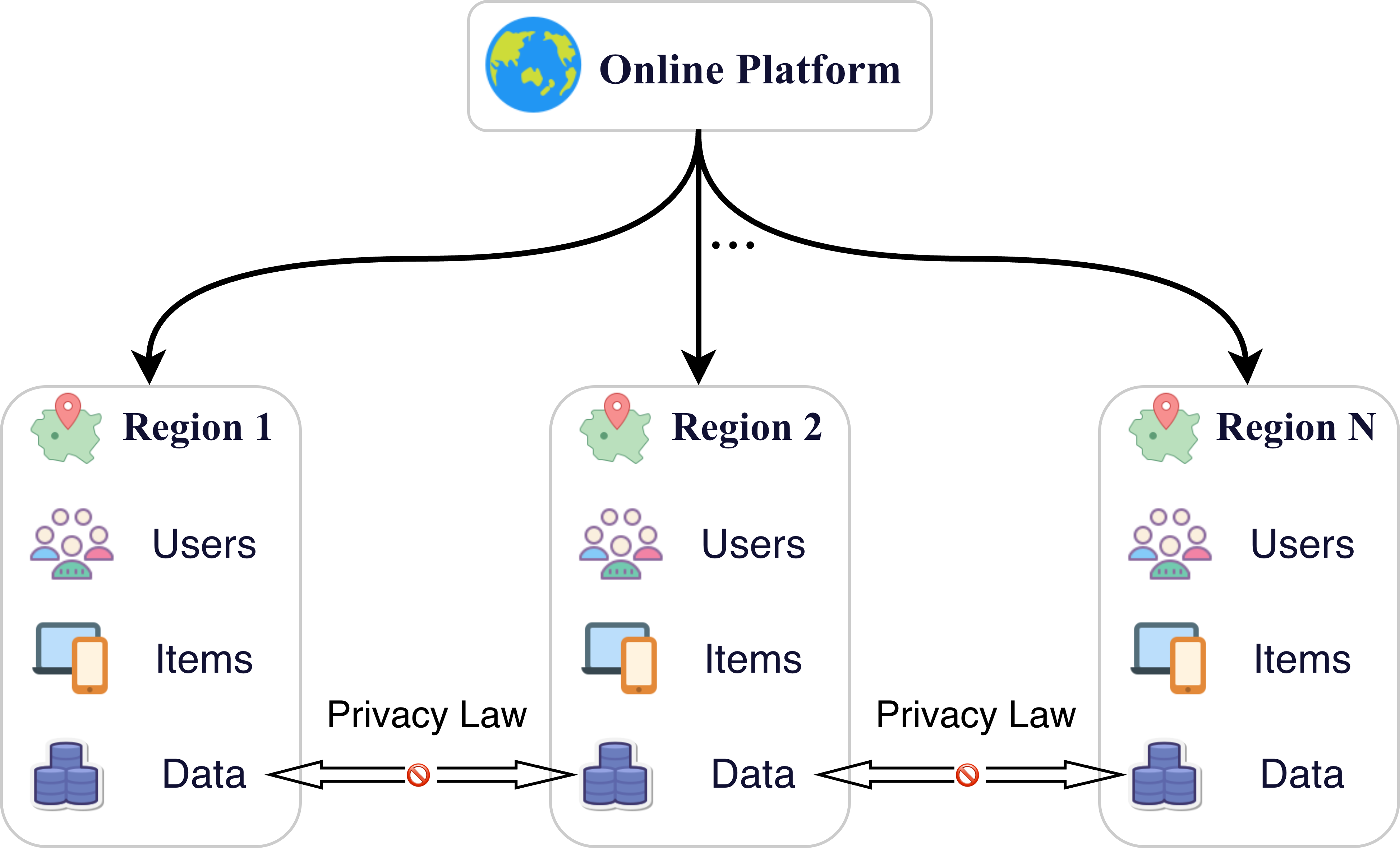}
    \Description[<Figure 1. Fully described in the text.>]{<A full description of Figure 1 can be found in the first and second paragraphs of Section 1.>}
\caption{A typical multi-market recommendation scenario. The online platform operates across multiple distinct regions. Due to strict privacy laws and data sovereignty regulations, the data are isolated within each region.}
\label{fig:intro}
\end{figure}
Existing MMR methods~\cite{xmarket,m3rec,ma,bert4xmr,MultiFS} mainly rely on centralized training. These approaches typically train a unified market-agnostic model on centralized data from all markets to capture shared knowledge, which is subsequently adapted to specific markets. Fundamentally, these works rely on the assumption that platforms can freely aggregate data from all markets into a central repository. However, this assumption is becoming invalid. As shown in Figure~\ref{fig:intro}, with the growing global emphasis on digital sovereignty and user privacy, numerous countries have enacted stringent laws and regulations, such as the GDPR\footnote{\url{https://gdpr-info.eu}}, to strictly limit the cross-border flow of personal data. Consequently, aggregating raw user behavior data into a centralized server for unified training is severely restricted, rendering traditional centralized MMR approaches infeasible in practice.

Federated learning has emerged as a promising approach to address privacy concerns in recommender systems~\cite{fed1,fed2,fed3}. By transmitting and aggregating model parameters while keeping sensitive user behavior data locally, clients can leverage shared knowledge to achieve superior generalization performance~\cite{fedavg,fedprox,fednova,fedcf}. However, applying this paradigm to multi-market settings presents unique challenges. Users across countries exhibit significant heterogeneity in cultural backgrounds and consumption habits, resulting in distinct patterns of interest. Simply merging these distinct market-specific patterns into a global model often leads to negative transfer, in which the unique characteristics of a given market are lost. Instead of mutually enhancing performance, such a ``one-size-fits-all'' aggregation mechanism introduces noise and conflicting signals into the global model parameters, corrupting the shared knowledge. Furthermore, embedding-based aggregation approaches~\cite{fedncf,fedmf,feddual,fedcia} are difficult to implement in practice. Due to disjoint ID spaces across markets, there is no unified mapping mechanism for items across regions. Consequently, directly aggregating embedding parameters is ineffective, as the model cannot align the item representations without a shared ID space.

To address these issues, we propose a novel federated multi-market click-through rate (CTR) prediction framework, named FedMM. Our method fuses collaborative information from distributed markets to construct global collaborative signals, which are then quantized into tokens and used as additional discrete features to augment each market's CTR performance. Specifically, the workflow operates as follows: each market first initializes a local collaborative filtering (CF) model (e.g., LightGCN~\cite{lightgcn}) to extract latent feature representations of users and items. Subsequently, we employ a dual-layer residual quantized variational autoencoder (RQ-VAE)~\cite{RQVAE} to quantize these local representations. The first-layer codebook is constructed in a federated manner by aggregating across all participating markets to capture global collaborative signals (shared patterns). The second layer codebook is learned locally to preserve market-specific characteristics, thereby encoding local collaborative signals. Upon completion of this quantization process, we obtain a sequence of semantic tokens enriched with hierarchical collaborative information. These tokens are then treated as discrete features and fed into the downstream CTR model. Training on these semantically meaningful tokens facilitates explicit interactions between global commonalities and local specificities, providing the CTR model with richer context. From a privacy perspective, the first-layer codebooks exchanged during the federated process function as semantic prototypes. This effectively obfuscates individual identities by sharing only abstract feature patterns, thereby ensuring robust privacy protection.

To summarize, this work makes the following main contributions:
\begin{itemize}[leftmargin=*]
    \item We propose FedMM, a privacy-preserving framework that introduces federated collaborative signal quantization to MMR. It is the first work to use the codebook as a medium for enabling privacy-preserving transmission and aligning disjoint ID spaces across markets.
    \item We design a coarse-to-fine quantization strategy based on a two-layer RQ-VAE to inject hierarchical global-local collaborative signals into the CTR model.
    \item FedMM provides a scalable solution to serve multiple markets simultaneously, avoiding the need for source-target pairing.
    \item We conduct extensive experiments on multiple real-world datasets involving varying numbers of markets (ranging from 2 to 4). Experimental results across various backbone models demonstrate the superiority of the proposed FedMM framework.
\end{itemize}

%% file: section/relatedwork.tex
\section{Related Work}\label{sec:related}
This section provides a brief overview of related work on two key research topics: multi-market recommendation and federated learning.

\subsection{Multi-Market Recommendation}
Multi-market recommendation (MMR) aims to leverage cross-market knowledge to enhance overall recommendation performance. The Meta-Learning paradigm has been widely adopted to adapt to distinct market distributions. For instance, FOREC~\cite{xmarket} utilizes model-agnostic meta-learning (MAML) to learn globally shared initialization parameters, which are subsequently fine-tuned with market-specific layers. Following this, hybrid strategies~\cite{mamlcf} have been proposed to combine MAML with traditional collaborative filtering to address underfitting for outlier users.
In parallel, some approaches focus on aligning embedding spaces. Methods such as MA~\cite{ma} and Bert4XMR~\cite{bert4xmr} introduce market-specific embeddings to map universal item representations into market-specific vector spaces. Additionally, M\(^3\)Rec~\cite{m3rec} employs random walk algorithms to capture high-order correlations across different markets.
Despite these advancements, the aforementioned methods generally overlook data privacy protection in markets. To address this, CAT-SR \cite{catsr} attempts to preserve privacy by transferring self-attention module weights. However, CAT-SR relies on the assumption of a data-rich source market to distill knowledge and is restricted to pairwise transfer. This limitation makes it difficult to implement in realistic multi-market scenarios.

\subsection{Federated Learning}
Federated learning (FL) is a distributed machine learning paradigm that enables collaborative model training across decentralized clients (e.g., edge devices or local servers) while keeping raw data localized to ensure privacy compliance~\cite{fedbat,panaceas,masked}. Typically, FL relies on aggregating model parameters to distill knowledge, aiming to learn a global model with strong generalization capabilities~\cite{fedavg}.
However, statistical heterogeneity across clients often leads to ``client drift.'' To mitigate this, FedProx~\cite{fedprox} and SCAFFOLD~\cite{scaffold} introduce proximal terms or control variates, respectively, to correct local updates. To further adapt to diverse local distributions, personalized federated learning has been proposed. For instance, FedAMP~\cite{fedamp} employs an attention mechanism to measure correlations among clients, enabling the aggregation of personalized parameters. Similarly, FedPer~\cite{fedper} achieves client-level personalization by decoupling the model into shared global representation layers and personalized local heads.
Although these FL methodologies can be applied to privacy-preserving multi-market recommendation (MMR), they often struggle in practice. Due to severe market heterogeneity, directly deploying these methods often leads to negative transfer, in which conflicting knowledge from distinct markets degrades local recommendation performance.

%% file: section/formulation.tex
\section{Problem Formulation}\label{sec:problem}
Let \( \mathcal{M}=\{M_1,M_2,\dots,M_K\} \) denote a set of \(K\) distinct markets (e.g., countries or regions) served by the global online platform. For the \(k\)-th market \(M_k\), the local dataset is denoted as \( \mathcal{D}_k=\{(u,v,y,\mathbf{x}) \mid u\in \mathcal{U}_k,v\in \mathcal{V}_k\} \), where \(\mathcal{U}_k\) and \(\mathcal{V}_k\) represent the sets of users and items in market \(M_k\), respectively. Each instance consists of a user \(u\), an item \(v\), a binary label \( y \in \{ 0,1 \} \) indicating whether a click event occurred, and a feature vector \(\mathbf{x}\) (containing user profiles and item attributes).

\noindent\textbf{Privacy Constraint.} Due to data sovereignty regulations and privacy concerns, raw data \(\mathcal{D}_k\) is stored on the local server node of market \(M_k\) and cannot be transferred to a central server or other markets. Only model parameters, gradients, or abstract representations are allowed to be exchanged under privacy-preserving protocols (e.g., local differential privacy).

\noindent\textbf{Disjoint ID Spaces.} We assume that the user populations across different markets are strictly non-overlapping. Regarding items, although partial overlaps may exist in the real-world merchandise (e.g., global products sold in multiple regions), their identifiers are encoded independently within each market. This means the same product corresponds to distinct, unrelated ID tokens in different markets. Formally, we treat both user and item ID spaces as disjoint across markets. This setup excludes simple ID-matching or embedding sharing strategies, as there are no common identifiers to directly bridge the diverse markets.

\noindent\textbf{Learning Objective.} The goal of the multi-market CTR prediction task is to collaboratively train a global model or a set of enhanced local models without sharing raw data, such that the prediction performance in each market is maximized. Specifically, for each instance in market \(M_k\), we aim to learn a function \(f_\Theta(\cdot)\) to predict the probability of a user clicking on an item:
\begin{equation}\label{eq:1}
\begin{aligned}
    \hat{y}_{uv} = f_\Theta(u,v,\mathbf{x}),
\end{aligned}
\end{equation}
where \(\Theta\) denotes the learnable parameters. The overall objective is to minimize the log-loss (binary cross-entropy) across all \(K\) markets:
\begin{equation}\label{eq:2}
\begin{aligned}
    \mathcal{L} = -\sum_{k=1}^{K}\sum_{(u,v,y,\mathbf{x}) \in \mathcal{D}_k} \left[y\text{log}(\hat{y}_{uv}) + (1-y)\text{log}(1-\hat{y}_{uv})\right].
\end{aligned}
\end{equation}

%% file: section/theoretica.tex
\section{Theoretical Analysis}\label{sec:theory}
In this section, we provide a theoretical analysis for the validity of aggregating locally updated codebooks, demonstrating that our FedMM paradigm approximates quantization on the global dataset.
\begin{assumption}
    All clients are initialized with the same global codebook. This means the $j$-th codeword initially represents the same semantic cluster across all local clients.
\end{assumption}
\begin{theorem}
    The aggregated codebook \(\mathcal{C}_{\text{fed}}\) derived from locally updated codebooks is mathematically equivalent to a codebook updated on the union of all local datasets.
\end{theorem}
\begin{proof}\let\qed\relax
We formulate the quantization process as finding a set of discrete vectors (i.e., a codebook), denoted as \(\mathcal{C}=\{\mathbf{c}_j\}_{j=1}^T\), that minimizes the Euclidean distance between input vectors and their nearest neighbor codewords. Let \(\mathcal{E}\) be the global set of vectors to be quantized, comprising \(K\) subsets, \(\mathcal{E}_1,\dots,\mathcal{E}_K\), distributed across \(K\) clients. The global optimization objective is to minimize the quantization error:
\begin{equation}\label{eq:3}
    \min_{\mathcal{C}}\mathcal{L}_{global}=\sum_{k=1}^K\sum_{\mathbf{e}\in\mathcal{E}_k}\lVert\mathbf{e}-q(\mathbf{e};\mathcal{C})\rVert^2,
\end{equation}
where \(q(\mathbf{e};\mathcal{C})=\arg\min_{\mathbf{c}\in\mathcal{C}}\lVert\mathbf{e}-\mathbf{c}\rVert^2\) denotes the nearest neighbor search operation. Suppose we employ gradient descent to update the codebook. For a specific codeword \(\mathbf{c}_j\), its gradient with respect to the global loss function is given by:
\begin{align}\label{eq:4}
    \nabla_{\mathbf{c}_j}\mathcal{L}_{global} & = \sum_{k=1}^K\sum_{\mathbf{e}\in\mathcal{E}_k,q(\mathbf{e})=\mathbf{c}_j}-2(\mathbf{e}-\mathbf{c}_j) \nonumber \\
      &= \sum_{k=1}^K\sum_{\mathbf{e}\in\mathcal{E}_k,q(\mathbf{e})=\mathbf{c}_j}2(\mathbf{c}_j-\mathbf{e}).
\end{align}
Let \(\mathcal{L}_{local}^k\) denote the local loss function of the \(k\)-th client. Similarly, the gradient of the local quantization loss is:
\begin{equation}\label{eq:5}
\nabla_{\mathbf{c}_j}\mathcal{L}_{local}^k=\sum_{\mathbf{e}\in\mathcal{E}_k,q(\mathbf{e})=\mathbf{c}_j}2(\mathbf{c}_j-\mathbf{e}).
\end{equation}
Leveraging the linearity of gradients, we obtain:
\begin{equation}\label{eq:6}
\nabla_{\mathbf{c}_j}\mathcal{L}_{global}=\sum_{k=1}^K\nabla_{\mathbf{c}_j}\mathcal{L}_{local}^k.
\end{equation}
This implies that the gradient used to optimize the global objective equals the sum of all local gradients. This mathematical equivalence provides the theoretical foundation for our federated codebook aggregation framework, demonstrating how it captures and integrates shared collaborative patterns across diverse clients. 
\end{proof}

%% file: section/method.tex
\section{The Proposed Framework}\label{sec:method}
In this section, we elaborate on the proposed federated collaborative signal quantization framework for multi-market CTR prediction, which is called FedMM. The overall architecture of FedMM is illustrated in Figure~\ref{fig:framework}. The presentation proceeds as follows: first, we describe the local collaborative signal learning process; second, we introduce the codebook-based federated aggregation paradigm to capture cross-market global signals; and finally, we describe the utilization of collaborative tokens for local CTR enhancement and the optimization objectives.
\begin{figure*}[htbp]
    \centering
    \includegraphics[width=1.\textwidth]{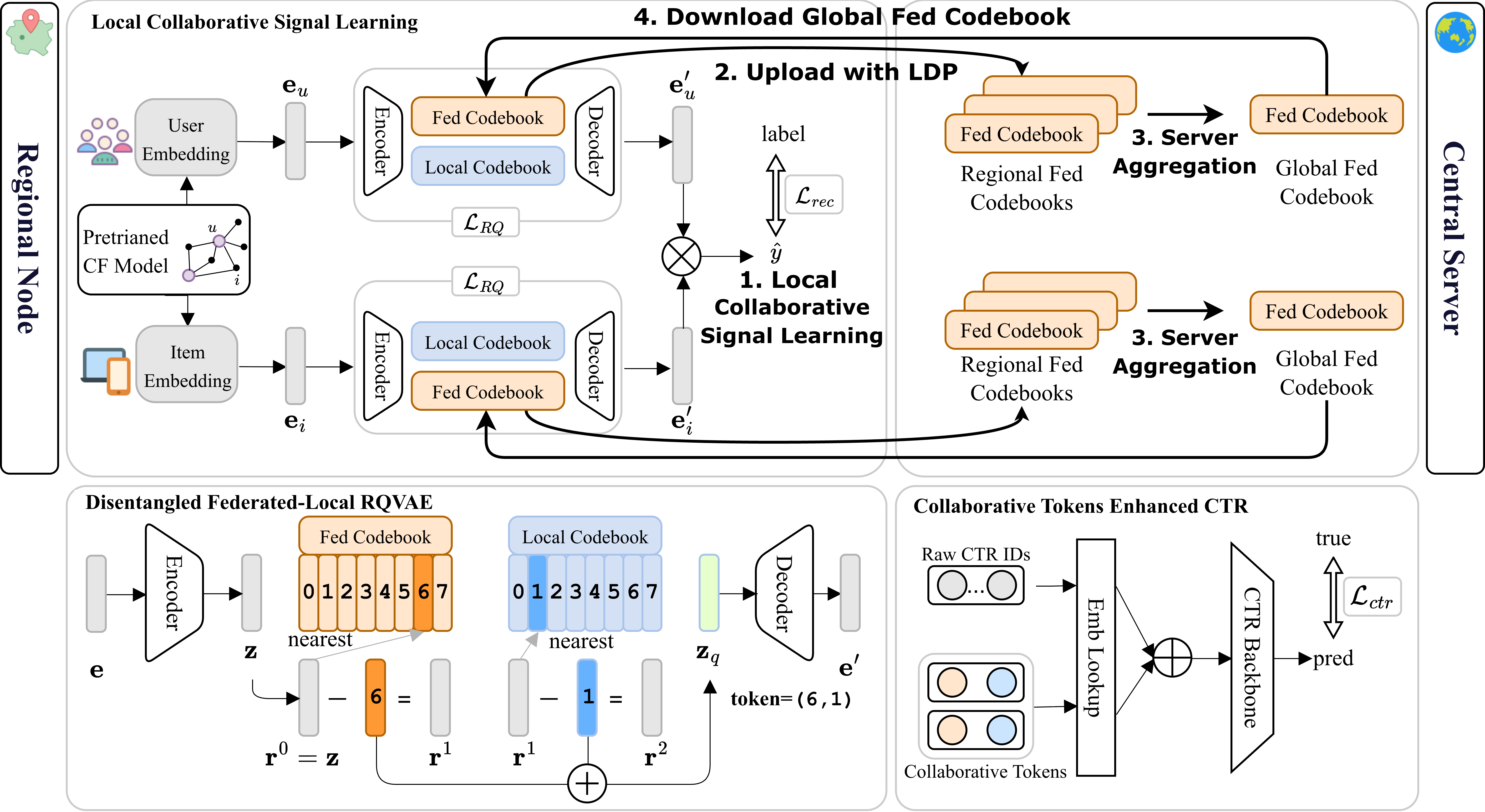}
    \Description[<Figure 2. Fully described in the text.>]{<A full description of Figure 2 can be found in Section 5.>}
    \caption{Overview of our FedMM framework. The framework comprises three key components: (1) Local Signal Learning: regional nodes learn to disentangle user/item representations using the disentangled federated-local RQ-VAE (zoomed in bottom-left). (2) Federated Aggregation: regional nodes upload federated codebooks protected by local differential privacy (LDP), which are then aggregated by the central server to synchronize global codebooks (top-right). (3) CTR Enhancement: in the downstream task (bottom-right), the discrete collaborative tokens are employed to enhance the CTR backbone model.}
    \label{fig:framework}
\end{figure*}
\subsection{Local Collaborative Signal Learning}\label{subsec:local}
This subsection details the generation of local collaborative tokens. We proceed under the assumption that each regional node has already acquired collaborative embeddings for users and items via a pre-trained collaborative filtering (CF) model. For instance, using LightGCN~\cite{lightgcn} can capture collaborative signals by performing weighted aggregation over adjacent neighbors. Inspired by generative recommender systems~\cite{generative}, we employ a residual quantized variational autoencoder (RQ-VAE)~\cite{RQVAE} to quantize these embeddings into discrete collaborative tokens. This mechanism injects hierarchical collaborative information into the tokens, thereby facilitating enhanced CTR prediction performance.

RQ-VAE serves as a multi-stage quantizer, employing residual quantization to generate fine-grained discrete representations. An RQ-VAE with \(m\) levels consists of \(m\) ordered codebooks. For any given input vector, the quantization process yields a discrete code sequence of length \(m\), where each entry represents the index of the selected codeword in the corresponding codebook. 
In the FedMM framework, we deploy two parallel RQ-VAE streams to process user and item embeddings independently. The hierarchical nature of RQ-VAE is inherently suited for multi-granularity quantization. 
Specifically, the first-layer user codebook \(\mathcal{C}_{\text{fed}}^U\) (and analogously \(\mathcal{C}_{\text{fed}}^I\) for item) is constructed via federated aggregation from all regional nodes to capture coarse-grained global collaborative signals, while the second-layer local codebook \(\mathcal{C}_{\text{local}}^U\) is learned locally to capture fine-grained local collaborative signals.

In the following, we detail the quantization process using user collaborative embeddings as an example, noting that the procedure for item embeddings is analogous. Given a user collaborative embedding \(\mathbf{e}_u\), it is first fed into an encoder to be mapped into the latent space, yielding a latent representation \(\mathbf{z}_u\). We treat \(\mathbf{z}_u\) as the $0$th-level residual, i.e. \(\mathbf{r}^0_u = \mathbf{z}_u\). Subsequently, \(\mathbf{r}_u^0\) is quantized by retrieving the nearest embedding \(\mathbf{c}^0_u\) from the size-\(T\) federated codebook \(\mathcal{C}_{\text{fed}}^U=\{\textbf{c}_{j}^0\}_{j=1}^T\),
\begin{equation}\label{eq:7}
    \begin{gathered}
    \mathbf{z}_u=E(\mathbf{e}_u),\\
    \mathbf{c}^0_u= \arg\max_{\mathbf{c}_j^0 \in \mathcal{C}_{\text{fed}}^U} \;P(\mathbf{c}_j^0 \mid \mathbf{r}^0_u),\;\text{where}\; P(\mathbf{c}_j^0\mid \mathbf{r}^0_u)=\frac{e^{\text{sim}(\mathbf{r}^0_u,\mathbf{c}_j^0)/\tau}}{\sum_{k=1}^T e^{\text{sim}(\mathbf{r}^0_u,\mathbf{c}_k^0)/\tau}},
    \end{gathered}
\end{equation}
where the \(E(\cdot)\) is the encoder in RQ-VAE, \(\text{sim}(\cdot)\) denotes the similarity metric, \(\tau\) is the temperature parameter. We utilize cosine similarity for quantization to maintain consistency with mainstream CF models (e.g., MF~\cite{MF}, NGCF~\cite{NGCF}, LightGCN~\cite{lightgcn}), which inherently rely on inner products to model interactions. This alignment ensures that the discrete codes reflect the CF backbone's optimization goals~\cite{cogcl}. Furthermore, compared with Euclidean distance, cosine similarity is robust to variations in vector magnitude. This property is crucial for eliminating activity bias, enabling the quantization process to focus on capturing users' intrinsic semantic preferences.

After obtaining the $1$st-level residual via \(\textbf{r}_u^1=\textbf{r}_u^0-\textbf{c}_u^0\), we employ the local codebook for its quantization, following a procedure analogous to Eq.~\eqref{eq:7} to retrieve the nearest embedding \(\mathbf{c}_u^1\) in local codebook \(\mathcal{C}_{\text{local}}^U=\{\textbf{c}_{j}^1\}_{j=1}^T\). We optimize the codebooks under the guidance of the loss \(\mathcal{L}_{code}^U\). This objective maximizes the likelihood of the input embedding given its corresponding optimal codeword:
\begin{equation}\label{eq:8}
    \mathcal{L}_{code}^U = -[\log P(\mathbf{c}_u^0\mid \mathbf{z_u^0}) + \log P(\mathbf{c}_u^1\mid \mathbf{z_u^1})].
\end{equation}
Analogously, the loss for item embeddings is \(\mathcal{L}_{code}^I\). The codebook learning loss for the framework is the aggregate of the user and item terms: \(\mathcal{L}_{code}=\mathcal{L}_{code}^U+\mathcal{L}_{code}^I\).
We obtain the quantized vector \(\mathbf{z}_u^q\) by summing the quantization outputs from both the federated and local layers, 
\begin{equation}\label{eq:9}
    \mathbf{z}_u^q = \mathbf{c}_u^0 + \mathbf{c}_u^1.
\end{equation}
To preserve the original collaborative semantics and prevent mode collapse, we introduce a semantic alignment loss \(\mathcal{L}_{al}\), which leverages the fixed input distribution as a semantic anchor:
\begin{equation}\label{eq:10}
    \mathcal{L}_{al}^U = -\log P(\mathbf{e}_u \mid \mathbf{z}_u^q), \;\text{where}\; P(\mathbf{e}_u\mid \mathbf{z}^q_u)=\frac{e^{\text{sim}(\mathbf{z}^q_u,\mathbf{e}_u)/\tau}}{\sum_{j=1}^B e^{\text{sim}(\mathbf{z}^q_u,\mathbf{e}_{u,j})/\tau}},
\end{equation}
where \(B\) denotes the batch size. To jointly optimize for codebook updates and semantic alignment, we integrate these terms into a unified quantization objective \(\mathcal{L}_{RQ}\):
\begin{equation}\label{eq:11}
    \mathcal{L}_{RQ} = \mathcal{L}_{code} + \mathcal{L}_{al}.
\end{equation}
Then, we input \(\mathbf{z}_u^q\) into the decoder  \(D(\cdot)\) to reconstruct the original collaborative embedding,
\begin{equation}\label{eq:12}
    \mathbf{e}_u'=D(\mathbf{z}_u^q).
\end{equation}

After obtaining the reconstructed embeddings for both the user and the item, we compute their inner product to evaluate the user-item co-occurrence patterns, optimizing the framework via the reconstruction loss \(\mathcal{L}_{rec}\):
\begin{equation}\label{eq:13}
    \begin{gathered}
    \hat{y}={\mathbf{e}_u'}^\top \mathbf{e}_i',\\
    \mathcal{L}_{rec} = -\left[y\text{log}(\hat{y}) + (1-y)\text{log}(1-\hat{y})\right].
    \end{gathered}
\end{equation}
The local learning process in FedMM is optimized by a composite objective function:
\begin{equation}\label{eq:14}
    \mathcal{L}_{total} = \mathcal{L}_{rec}+ \lambda\mathcal{L}_{RQ},
\end{equation}
where \(\lambda\) controls the weight of the loss term.

\subsection{Codebook-based Federated Collaborative Signal Aggregation}
\subsubsection{Client Training \& Upload.}
At the beginning of each federated round, every regional node executes \(E_\text{local}\) epochs of the local collaborative signal learning process (as detailed in Section~\ref{subsec:local}). Upon completion, the nodes upload both the user and item federated codebooks to the central server. Prior to transmission, each client applies a local differential privacy (LDP) mechanism to inject random perturbations into the codebooks, ensuring data privacy.
\subsubsection{Server Aggregation.}
Upon receiving the uploaded codebooks from all regional nodes, the central server executes an aggregation to generate the global federated codebooks. Specifically, the aggregation is performed in an index-aligned manner to preserve semantic consistency. For each index \(j\) in the codebook, the global codeword \(\mathbf{c}_{j}^0\) is computed by aggregating the corresponding set of local codewords \(\{\mathbf{c}_{j,k}^0\}_{k=1}^K\) from all \(K\) participating clients. This process integrates the collaborative signals from all distinct regions into a unified global representation. 
\subsubsection{Client Download \& Adaptation.}
After each regional node downloads and loads the updated global federated codebook, a discrepancy may arise between the new global priors and the existing local codebook. To facilitate the adaptation of the local codebook to these global signals, we introduce a local codebook alignment step, where we perform \(E_\text{adapt}\) rounds of adaptation on the local codebook. 

\subsubsection{Privacy Analysis.}
The privacy of FedMM is protected through a three-layer defense mechanism that leverages the characteristics of our layered quantization framework:
\begin{itemize}[leftmargin=*]
\item \textbf{Prototype-based Anonymization.} The transmission in FedMM is restricted to the Federated Codebook, which consists of discrete semantic prototypes rather than raw user data or dense gradients. Since each prototype represents a statistical aggregate of massive embeddings, it inherently obfuscates individual identities. This vector quantization process serves as a compression function, filtering out entity-specific identification details while preserving common semantic structures.
\item \textbf{Granularity Isolation.} We strictly isolate fine-grained user information by keeping the Local Codebook and residual representations on local devices. Specifically, only global patterns are exchanged. The residual vectors, which capture the unique and detailed characteristics of market features, are never exposed to the server. This architecture creates an information gap where the transmitted global prototypes serve only as rough approximations. Without access to the locally stored residuals, accurately recovering the original embeddings is practically impossible, thereby preventing the reconstruction of raw data.
\item \textbf{Local Differential Privacy.} To further immunize the system against inverse reconstruction attacks, we apply local differential privacy (LDP) via the Laplace perturbation mechanism. By introducing zero-mean Laplacian noise into the codebooks before uploading, we ensure that the global server cannot distinguish whether a specific update originated from any single user or item, providing a formal mathematical guarantee of privacy.
\end{itemize}

\subsection{Collaborative Tokens Enhanced CTR}
Upon the convergence of the federated training process, we proceed to the token generation phase. In this step, each regional node performs a single pass of quantization (in inference mode without gradient computation) on its local collaborative embeddings to generate the corresponding collaborative tokens \(\mathbf{Z}^k\). They act as discrete carriers of hierarchical collaborative signals. We integrate these tokens into the downstream CTR prediction task to augment the feature space. The complete FedMM training pipeline is detailed in Algorithm~\ref{alg:FedMM}.
\begin{algorithm}[htbp]
\caption{FedMM Training Pipeline}
\label{alg:FedMM}
\begin{algorithmic}[1]
\Statex \textbf{Phase 1: Federated Codebook Training}
\Statex \textbf{Input} Pre-trained collaborative embeddings $\mathcal{E}_k$ for each regional node $k$
\Statex \textbf{Parameters} Federated Rounds $R$, Local Epochs $E_{\text{local}}$, Adaptation Epochs $E_{\text{adapt}}$
\State Initialize global fed codebooks $\mathcal{C}_{\text{fed}}=\{\mathcal{C}_{\text{fed}}^U,\mathcal{C}_{\text{fed}}^I\}$
\For{each federated round $r = 1, 2, \dots, R$}
    \State Server broadcasts global codebook $\mathcal{C}_{\text{fed}}$ to all $K$ regional nodes
    \For{each regional node $k \in \{1, \dots, K\}$ \textbf{in parallel}}
        \For{each adaptation epoch $j = 1, \dots, E_{\text{adapt}}$}
            \State Update $\mathcal{C}_\text{local}^k$ to minimize the loss \(\mathcal{L}_{total}\)
        \EndFor
        \For{each local epoch $e = 1, \dots, E$}
            \State Train local quantizer by \(\mathcal{L}_{total}\)
        \EndFor
        \State Apply Local Differential Privacy: $\tilde{\mathcal{C}}_{\text{fed}}^{k} \gets \text{LDP}(\mathcal{C}_{\text{fed}}^{k}, \epsilon)$
        \State Upload $\tilde{\mathcal{C}}_{\text{fed}}^{k}$ to server
    \EndFor
    \State Server aggregates: $\mathcal{C}_{\text{fed}} \gets \mathrm{ServerAgg}\left( \{ \tilde{\mathcal{C}}_{\text{fed}}^{k} \}_{k=1}^K \right)$
\EndFor
\Statex \textbf{Phase 2: Local CTR Prediction Enhancement}
\Statex \textbf{Input} Local data $\mathcal{D}_k$ for each regional node $k$
\Statex \textbf{Parameters} CTR Model Training Epochs $E_\text{ctr}$
    \For{each regional node $k \in \{1, \dots, K\}$ \textbf{in parallel}}
        \State Freeze the learned codebooks $\mathcal{C}_{\text{fed}}$ and $\mathcal{C}_{\text{local}}^k$
        \State \textbf{Quantize} local collaborative embeddings to obtain collaborative tokens: 
        \State \quad $\mathbf{Z}_k \gets \text{Quantize}(\mathcal{E}_k; \mathcal{C}_{\text{fed}}, \mathcal{C}_\text{local}^k)$
        \For{each epoch $e = 1, \dots, E_{\text{ctr}}$}
            \For{each batch $(u,v,y,\mathbf{x})$ in local data $\mathcal{D}_k$}
                \State Retrieve collaborative tokens from $\mathbf{Z}_k$
                \State Train CTR model parameters $\Theta_k$ by minimizing $\mathcal{L}_{\text{ctr}}$
            \EndFor
        \EndFor
    \EndFor
\end{algorithmic}
\end{algorithm}

Specifically, for each user (or item), the generated token pair \((id_{\text{fed}},id_{\text{local}})\) is treated as two additional feature fields. We employ a learnable embedding layer to map these discrete indices into dense vectors, denoted as \(\mathbf{e}_{\text{fed}}\) and \(\mathbf{e}_{\text{local}}\). As illustrated in the bottom-right of Figure~\ref{fig:framework}, these token embeddings are concatenated with the embeddings of raw CTR features. The fused feature vector \(\mathbf{x}_{\text{fused}}\) is then fed into the CTR backbone model: 
\begin{equation}\label{eq:15}
    \hat{y}_{ctr} = f_{\Theta_k}(\mathbf{x_{\text{fused}}}),
\end{equation} 
where \(\Theta_k\) denotes the parameters of the CTR model for regional node \(k\). It is worth noting that the generated collaborative tokens are model-agnostic. These discrete signals serve as universal semantic features that can be integrated into any mainstream deep CTR architecture, such as DeepFM, DCN, or DNN, to enhance predictive performance. The final objective is to minimize the log-loss:
\begin{equation}\label{eq:16}
\begin{aligned}
    \mathcal{L}_{ctr} = -\sum_{(u,v,y,\mathbf{x}) \in \mathcal{D}_k} \left[y\text{log}(\hat{y}_{ctr}) + (1-y)\text{log}(1-\hat{y}_{ctr})\right].
\end{aligned}
\end{equation}

%% file: section/experiments.tex
\section{Experiments}\label{sec:experiments}
In this section, we conduct a series of experiments to answer the following four critical questions. The source code is publicly available at: \url{https://github.com/JunZhangJz/FedMM}.
\begin{itemize}[leftmargin=*]
    \item RQ1: How does our FedMM perform compared with the baselines?
    \item RQ2: What are the contributions of the key components in our FedMM?
    \item RQ3: What are the effects of different settings of key parameters?
    \item RQ4: What is the role of the learned federated codebook in FedMM?
\end{itemize}
\subsection{Experiment Setup}\label{sec:experiments:setup}
\subsubsection{Datasets.}
In this study, we utilize the XMarket\footnote{\url{https://xmrec.github.io/}} dataset, a large-scale cross-market e-commerce dataset collected by Bonab et al.~\cite{xmarket} via crawling 18 distinct Amazon marketplaces spanning 11 languages. To simulate multi-market recommendation scenarios of varying scales and complexities, we constructed three datasets from XMarket comprising 2, 3, and 4 markets, respectively (denoted as 2-market, 3-market, and 4-market). Detailed statistics of these datasets are summarized in Table~\ref{tab:datasets}.

\begin{table}[htbp]
\centering
\caption{Dataset statistics.}
\resizebox{0.48\textwidth}{!}{
\begin{tabular}{c|cc|ccc|cccc}
\specialrule{0.1em}{1pt}{1pt}
    \multicolumn{1}{c|}{\textbf{Dataset}} & \multicolumn{2}{c|}{\textbf{\textsf{2-Market}}} & \multicolumn{3}{c|}{\textbf{\textsf{3-Market}}} & \multicolumn{4}{c}{\textbf{\textsf{4-Market}}} \\
    \hline
    \multicolumn{1}{c|}{\textbf{Market}} & uk & us  & es & jp & mx & ca & de & fr & it \\
    \specialrule{0.1em}{0pt}{1pt}
        \#\textbf{U} & 152k & 432k & 13k & 5k & 9k & 35k & 25k & 19k & 32k \\
        \#\textbf{I} & 60k & 83k & 4k & 2k & 5k & 13k & 5k & 4k & 5k \\
        \#\textbf{R} & 1.3m & 3.9m & 72k & 27k & 51k & 210k & 134k & 110k & 191k \\
\specialrule{0.1em}{1pt}{1pt}
\end{tabular}}
\label{tab:datasets}
\end{table}
For data preprocessing, we retained users with more than three interactions. For the pre-training phase of the collaborative filtering (CF) model, we adopted the leave-one-out strategy for data splitting. For the click-through rate (CTR) prediction task, we converted explicit ratings into binary feedback: ratings of 4 or higher were treated as positive interactions (i.e., clicks), while ratings below 4 were regarded as non-interactions. Following the setup of prior research~\cite{multiemb}, the dataset for the CTR task was partitioned into training, validation, and test sets in an 8:1:1 ratio.
\newcommand{\myhline}{\noalign{\vskip 1pt}\hline\noalign{\vskip 1pt}}
\begin{table*}[htpb]
\centering
\caption{Results on all datasets, where the best results are marked in bold. Note that $^{*}$ indicates a significance level of $p\leq 0.05$ based on a two-sample t-test between our method and the best baseline.}
\resizebox{1.0\textwidth}{!}{
\begin{tabular}{c|c|ccc|cccc|ccccc}
\specialrule{0.12em}{1pt}{0pt}
    \multicolumn{2}{c|}{\textbf{Dataset}} & \multicolumn{3}{c|}{\textbf{\textsf{2-Market}}} & \multicolumn{4}{c|}{\textbf{\textsf{3-Market}}} & \multicolumn{5}{c}{\textbf{\textsf{4-Market}}} \\
    \hline
    \multicolumn{2}{c|}{\textbf{Method}} & uk & us & Overall & es & jp & mx & Overall & ca & de & fr & it & Overall\\
    \specialrule{0.12em}{0pt}{1pt}
    \multirow{6}{*}{DNN} 
        & Local & 0.7916  & 0.7941  & 0.7929  & 0.8667  & 0.8714  & 0.8888  & 0.8756  & 0.8627  & 0.9225  & 0.9191  & 0.9045  & 0.9022  \\
        & FedAvg & 0.7913  & 0.7942  & 0.7928  & 0.8535  & 0.8739  & 0.8897  & 0.8724  & 0.8498  & 0.9205  & 0.9125  & 0.9029  & 0.8964  \\
        & FedPer & 0.7920  & 0.7931  & 0.7926  & 0.8581  & 0.8693  & 0.8902  & 0.8725  & 0.8454  & 0.9295  & 0.9216  & 0.9047  & 0.9003  \\
        & FedProx & 0.7919  & 0.7946  & 0.7933  & 0.8562  & 0.8604  & 0.8775  & 0.8647  & 0.8541  & 0.9207  & 0.9188  & 0.9024  & 0.8990  \\
        & pFedPara & 0.7902  & 0.7925  & 0.7914  & 0.8618  & 0.8747  & 0.8803  & 0.8723  & 0.8487  & 0.9086  & 0.9090  & 0.9019  & 0.8921 \\
        & FedMM & 0.7934  & 0.7957 & \textbf{0.7946*}  & 0.8736  & 0.8741  & 0.8935  & \textbf{0.8804*}  & 0.8652  & 0.9271  & 0.9240  & 0.9118  & \textbf{0.9070*}  \\
    \myhline
    \multirow{6}{*}{DeepFM} 
        & Local & 0.7915  & 0.7943  & 0.7929  & 0.8661  & 0.8726  & 0.8863  & 0.8750  & 0.8635  & 0.9222  & 0.9221  & 0.9047  & 0.9031  \\
        & FedAvg & 0.7917  & 0.7937  & 0.7927  & 0.8568  & 0.8675  & 0.8909  & 0.8717  & 0.8418  & 0.9169  & 0.9107  & 0.8919  & 0.8903  \\
        & FedPer & 0.7923  & 0.7931  & 0.7927  & 0.8598  & 0.8653  & 0.8895  & 0.8715  & 0.8388  & 0.9190  & 0.9099  & 0.8999  & 0.8919  \\
        & FedProx & 0.7916  & 0.7947  & 0.7932 & 0.8485  & 0.8961  & 0.8856  & 0.8767  & 0.8622  & 0.8991  & 0.8971  & 0.8866  & 0.8863  \\
        & pFedPara & 0.7902  & 0.7926  & 0.7914  & 0.8623  & 0.8740  & 0.8758  & 0.8707  & 0.8510  & 0.9127  & 0.9070  & 0.8912  & 0.8905 \\
        & FedMM & 0.7933  & 0.7961  & \textbf{0.7947*}  & 0.8755  & 0.8722  & 0.8956  & \textbf{0.8811*}  & 0.8612  & 0.9310  & 0.9254  & 0.9133  & \textbf{0.9077*}  \\
    \myhline
    \multirow{6}{*}{DCN} 
        & Local & 0.7916  & 0.7942  & 0.7929  & 0.8659  & 0.8736  & 0.8893  & 0.8763  & 0.8619  & 0.9225  & 0.9173  & 0.9079  & 0.9024  \\
        & FedAvg & 0.7916  & 0.7936  & 0.7926  & 0.8593  & 0.8703  & 0.8777  & 0.8691  & 0.8572  & 0.9256  & 0.9225  & 0.9037  & 0.9023  \\
        & FedPer & 0.7920  & 0.7929  & 0.7925  & 0.8709  & 0.8704  & 0.8815  & 0.8743  & 0.8582  & 0.9266  & 0.9167  & 0.9016  & 0.9008  \\
        & FedProx & 0.7919  & 0.7941  & 0.7930  & 0.8626  & 0.8794  & 0.8942  & 0.8787  & 0.8641  & 0.9280  & 0.9223  & 0.9071  & 0.9054  \\
        & pFedPara & 0.7900  & 0.7928  & 0.7914  & 0.8614  & 0.8689  & 0.8766  & 0.8690  & 0.8564  & 0.9040  & 0.9093  & 0.8982  & 0.8920 \\
        & FedMM & 0.7933  & 0.7957  & \textbf{0.7945*}  & 0.8698  & 0.8763  & 0.8985  & \textbf{0.8815*}  & 0.8680  & 0.9294  & 0.9217  & 0.9139  & \textbf{0.9083*}  \\
\specialrule{0.12em}{1pt}{1pt}
\specialrule{0.12em}{0pt}{1pt}
    \multirow{4}{*}{MMR}
        & MAML & 0.7916  & 0.7942  & 0.7929  & 0.8669  & 0.8764  & 0.8951  & 0.8795  & 0.8625  & 0.9257  & 0.9176  & 0.9076  & 0.9034  \\
        & FOREC & 0.7928  & 0.7949  & 0.7939  & 0.8680  & 0.8839  & 0.8939  & \textbf{0.8819}  & 0.8610  & 0.9284  & 0.9206  & 0.9127  & 0.9057  \\
        & MA & 0.7914  & 0.7939  & 0.7927  & 0.8667  & 0.8780  & 0.8933  & 0.8793  & 0.8638  & 0.9272  & 0.9261  & 0.9066  & 0.9059 \\
        & FedMM(DNN) & 0.7934  & 0.7957 & \textbf{0.7946}  & 0.8736  & 0.8741  & 0.8935  & 0.8804  & 0.8652  & 0.9271  & 0.9240  & 0.9118  & \textbf{0.9070*} \\
\specialrule{0.12em}{1pt}{1pt}
\end{tabular}}
\label{tab:main_result}
\end{table*}

\subsubsection{Metrics.}
We adopt the area under the ROC curve (AUC) as our evaluation metric, a widely used measure for assessing the performance of CTR prediction models. A higher AUC value indicates better performance. Notably, an improvement of more than 0.1\% in AUC is considered significant for the CTR prediction task~\cite{retrieval,ssim}.
\subsubsection{Backbones \& Baselines.}
To validate the generalization capability of FedMM, we integrate it with three backbone models: DNN~\cite{widedeep}, DeepFM~\cite{DeepFM}, and DCN~\cite{DCN}. We then comprehensively evaluate the performance of our proposed FedMM by comparing it against a diverse set of baselines spanning three categories: single-market local models, federated learning algorithms, and multi-market recommendation approaches. It is important to note that the multi-market baselines listed below assume centralized data availability and do not account for cross-market data privacy.

\noindent\textbf{Single-Market Local Models:}
\begin{itemize}[leftmargin=*]
\item \textbf{DNN}~\cite{widedeep}: A standard deep neural network that utilizes fully connected layers to capture non-linear interactions between user and item features.
\item \textbf{DeepFM}~\cite{DeepFM}: A hybrid model that integrates a factorization machine (FM) component for low-order interactions and a deep neural network for high-order feature interactions in an end-to-end manner.
\item \textbf{DCN}~\cite{DCN}: A model designed to efficiently learn bounded-degree feature interactions using a specialized cross network alongside a deep network, eliminating the need for manual feature engineering.
\end{itemize}
\textbf{Federated Learning Methods:}
\begin{itemize}[leftmargin=*]
\item \textbf{FedAvg}~\cite{fedavg}: The standard federated learning algorithm that aggregates local model updates via weighted averaging to derive a shared global model.
\item \textbf{FedPer}~\cite{fedper}: A personalized federated learning framework that decouples the model into shared base layers (aggregated globally) and personalized head layers (kept local), allowing adaptation to local data distributions.
\item \textbf{FedProx}~\cite{fedprox}: An optimization-based federated learning method designed to tackle statistical heterogeneity. It introduces a proximal term to the local objective function to limit the divergence of local updates from the global model.
\item \textbf{pFedPara}~\cite{fedpara}: A communication-efficient personalized federated learning approach that parameterizes layer weights using a low-rank Hadamard product. It decouples each layer's parameters into a global inner matrix aggregated on the server to learn shared knowledge and a local inner matrix kept on the client to capture specific personal features.
\end{itemize}
\textbf{Multi-Market Recommendation Methods:}
\begin{itemize}[leftmargin=*]
\item \textbf{MAML}~\cite{maml}: A model-agnostic meta-learning (MAML) approach that learns a set of highly improved initialization parameters, enabling the model to quickly adapt to specific market tasks with minimal gradient updates.
\item \textbf{FOREC}~\cite{xmarket}: An advancement over MAML that extends the meta-learning framework. It initializes the model with meta-learned parameters and subsequently appends market-specific layers for fine-tuning, thereby better capturing distinct market distributions.
\item \textbf{MA}~\cite{ma}: A method that models distinct markets by introducing market-specific embeddings, which project universal item representations into market-dependent vector spaces.
\end{itemize}

\subsubsection{Implementation Details.}
Regarding general hyper-parameter settings, all input features are discretized and mapped into a $16$-dimensional embedding space. The hidden layer sizes of the multi-layer perceptron (MLP) are fixed at $[512, 256, 128]$. We employ the Adam optimizer for training, with the learning rate tuned from the set \{1e-5, 5e-5, 1e-4, 5e-4, 1e-3, 5e-3\} and the \(l_2\) regularization coefficient selected from \{1e-6, 5e-6, 1e-5, 5e-5, 1e-4, 5e-4\}. We employ LightGCN~\cite{lightgcn} as our pre-trained collaborative filtering backbone. For the objective function, we empirically set \(\lambda=1\) to maintain a balanced magnitude between the reconstruction loss \(\mathcal{L}_{rec}\) and the quantization loss \(\mathcal{L}_{RQ}\). Regarding the baselines, we use publicly available open-source implementations to ensure reproducibility. For all federated learning methods, local differential privacy (LDP) is enforced by injecting noise sampled from \(Laplace(0, 0.001)\) into the parameters during the upload phase. Specifically for FedPer~\cite{fedper}, we designate the final output layer as the personalized component, while the preceding layers serve as the shared global representation. We set the proximal term coefficient \(\mu =0.1\) for FedProx~\cite{fedprox}, and the compression ratio to 0.5 for pFedPara~\cite{fedpara}. To ensure a fair comparison, all baseline methods are optimized via grid search within the same hyperparameter search space.

\subsection{RQ1: Overall Performance}\label{sec:experiments:rq1}
In this subsection, we integrate FedMM into DNN, DeepFM, and DCN to observe the performance improvements it brings. We also compare FedMM with representative federated learning algorithms and MMR methods. The overall performance of FedMM and the baseline methods is presented in Table~\ref{tab:main_result}. We have the following observations:
\begin{itemize}[leftmargin=*]
\item FedMM outperforms Local models across all datasets (2-Market, 3-Market, 4-Market) and backbone architectures (DNN, DeepFM, DCN). For instance, in the 4-Market scenario with DCN, FedMM achieves an AUC of 0.9083 compared to the Local model's 0.9024. By leveraging shared knowledge from other markets, FedMM enriches the local market's representational capabilities, leading to better generalization. 

\item FedMM significantly surpasses standard FL methods (FedAvg, FedProx, FedPer, pFedPara). Notably, we observe that FedAvg often performs worse than Local models, indicating a severe ``Negative Transfer'' phenomenon. While FedPer and FedProx mitigate this drop to some extent, they still fail to consistently exceed Local performance. This indicates that under severe market heterogeneity, directly aggregating model parameters across diverse markets introduces conflicting information, leading to the loss of market-specific features. Additionally, the limited expressive capacity inherent in pFedPara's low-rank-based parameterization results in a significant performance decline. In contrast, FedMM injects global collaborative signals into local models as complementary knowledge rather than a substitute. This architectural advantage empowers each market to selectively utilize beneficial information from the federation while preserving its local characteristics. Consequently, FedMM achieves a superior trade-off between global collaboration and local traits, avoiding the interference issues in standard FL methods.

\item FedMM achieves performance that is comparable to, and in many cases superior to, state-of-the-art centralized MMR methods. It is important to emphasize that MMR methods benefit from centralized training, enabling simultaneous access to all market data to optimize the joint distribution. FedMM achieves centralized-level performance without sharing raw data. It successfully captures higher-order cross-market correlations that are typically accessible only to centralized models, proving that privacy protection need not come at the cost of model performance.
\end{itemize}

\noindent \textbf{Communication Efficiency.}
We further analyze the communication efficiency of FedMM compared to standard federated learning approaches.
\begin{itemize}[leftmargin=*]
\item \textbf{Standard FL}: They are inherently burdened by the architecture of the deep recommendation backbone. These methods require synchronizing dense neural network layers, where the number of parameters scales with the model's depth and width. Even split-learning variants like FedPer still necessitate transmitting substantial shared bottom layers. Although pFedPara reduces per-round communication cost by transmitting low-rank parameters, it typically requires more rounds to converge, resulting in a higher total communication burden.
\item \textbf{FedMM}: FedMM adopts a lightweight communication strategy that completely decouples transmission from the backbone model size. Instead of uploading heavy model gradients, FedMM transmits only the compact federated codebook.
\end{itemize}

\begin{figure}[htbp] 
\centering
    \includegraphics[width=0.70\linewidth]{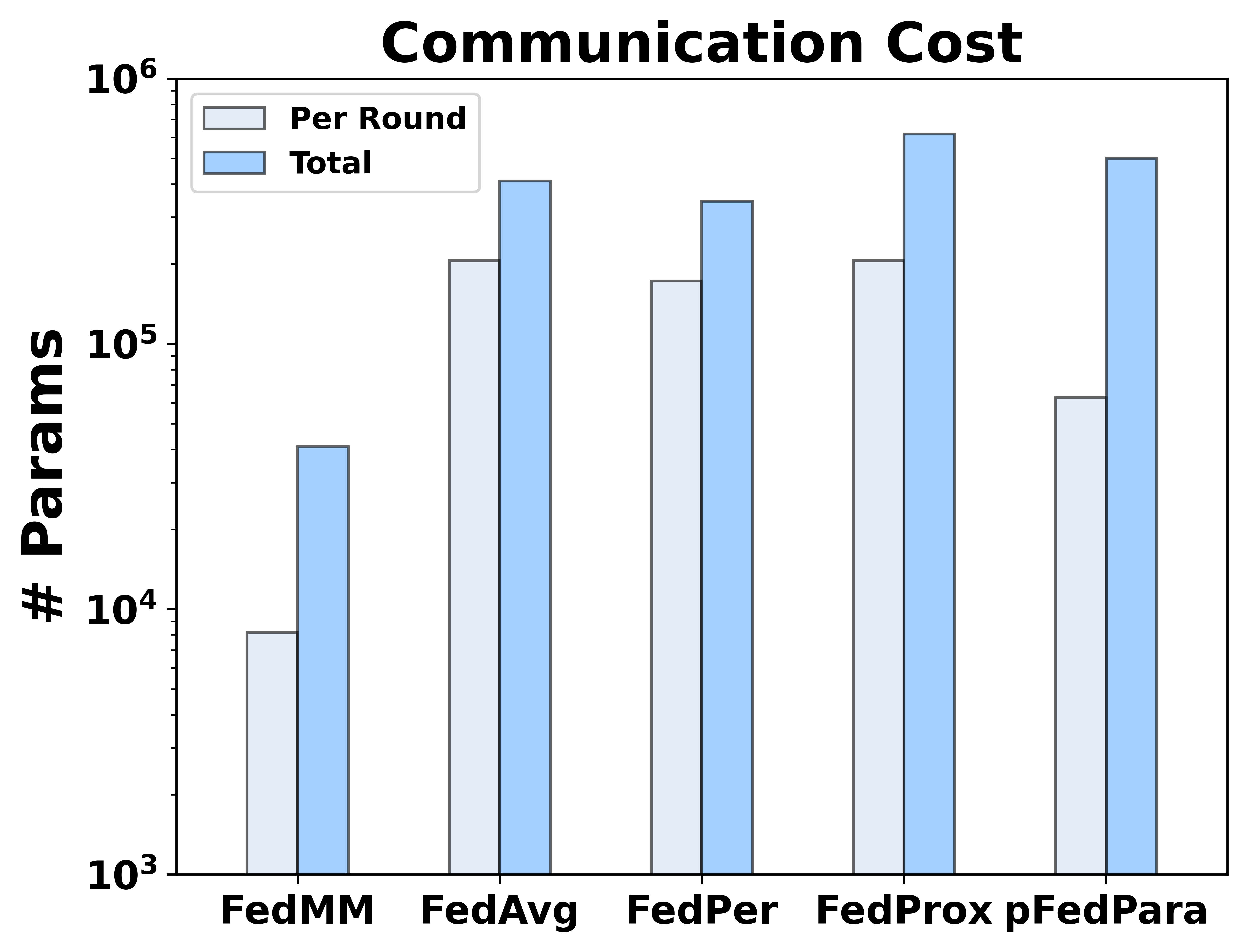}
    \Description[<Figure 3. Fully described in the text.>]{<A full description of Figure 3 can be found in the last paragraph of Section 6.2.>}
\caption{Communication cost comparison measured by the number of transmitted parameters per regional node (Log Scale). The chart contrasts the communication overhead per round (light blue) and for the total training process (dark blue).}
\label{fig:cost}
\end{figure}

This efficiency gap is clearly visualized in Figure~\ref{fig:cost}. The baseline methods incur high communication costs, reaching hundreds of thousands of parameters per round due to transmitting dense MLP layers. FedMM maintains a significantly lower profile. By avoiding transmission of the deep backbone entirely, FedMM reduces the communication payload by over 20$\times$ compared to the model-sharing baselines. This confirms that FedMM achieves superior efficiency by exchanging concise prototypes rather than heavy model weights.

\subsection{RQ2: Ablation Study}\label{sec:experiments:rq2}
To investigate the efficacy of the key components in FedMM, we construct three variants for comparison:
1) \textbf{w/o Local}. Here, the RQ-VAE is fully derived from federated training, in which all markets share a unified global codebook set, without preserving market-specific local codebooks. This setting tests the necessity of maintaining local characteristics in the face of market heterogeneity.
2) \textbf{w/o Global}. In this variant, the RQ-VAE component is trained exclusively on local data without any federated aggregation. This setting isolates the architectural benefit of the discrete variational structure from the collaborative knowledge transfer.
3) \textbf{Random}. A sanity check where the codebook embeddings are randomly initialized and frozen during training, validating whether the learned discrete representations effectively capture the collaborative signal.

As presented in Table~\ref{tab:ablation}, FedMM consistently outperforms all ablated variants across different datasets and backbone models, validating the synergy of its design. The performance degradation observed in the `w/o Global' variant underscores the critical role of cross-market collaboration; relying solely on local data limits the model's ability to generalize, particularly in scenarios plagued by data sparsity. Conversely, the `w/o Local' variant also exhibits suboptimal performance compared to FedMM. This decline suggests that the global codebook, while effective at capturing universal patterns, functions primarily as a coarse-grained reference. Without the local codebook to refine these general signals and capture market-specific residuals, the model lacks sufficient resolution to model the intricate user preferences unique to each market. Furthermore, the significant performance drop in the `Random Codebook' variant confirms that the improvement stems from the learned meaningful semantic prototypes in the codebook rather than mere structural complexity. In conclusion, FedMM achieves the optimal balance by injecting global collaborative signals to augment the representation capability of local models while retaining local features to prevent negative transfer, a trade-off that neither purely local nor purely global approaches can satisfy.

\begin{table}[htpb]
\centering
\footnotesize
\caption{Ablation Analysis on our FedMM, where the best results are marked in bold.}
\resizebox{0.45\textwidth}{!}{
\begin{tabular}{c|c|c|c|c}
\specialrule{0.12em}{1pt}{0pt}
    \multicolumn{2}{c|}{\textbf{Dataset}} & \multicolumn{1}{c|}{\textbf{\textsf{2-Market}}} & \multicolumn{1}{c|}{\textbf{\textsf{3-Market}}} & \multicolumn{1}{c}{\textbf{\textsf{4-Market}}} \\
    \specialrule{0.12em}{0pt}{1pt}
    \multirow{4}{*}{DNN} 
        & FedMM & \textbf{0.7946} & \textbf{0.8804} & \textbf{0.9070}  \\
        & w/o Local & 0.7937 & 0.8744 & 0.9039  \\
        & w/o Global & 0.7936 & 0.8755 & 0.9029  \\
        & Random & 0.7927 & 0.8713 & 0.8981  \\
    \myhline
    \multirow{4}{*}{DeepFM} 
        & FedMM & \textbf{0.7947} & \textbf{0.8811} & \textbf{0.9077}  \\
        & w/o Local & 0.7937 & 0.8787 & 0.9033  \\
        & w/o Global & 0.7927 & 0.8735 & 0.8990  \\
        & Random & 0.7928 & 0.8728 & 0.8987  \\
    \myhline
    \multirow{4}{*}{DCN} 
        & FedMM & \textbf{0.7945} & \textbf{0.8815} & \textbf{0.9083}  \\
        & w/o Local & 0.7935 & 0.8735 & 0.9073 \\
        & w/o Global & 0.7934 & 0.8754 & 0.9071  \\
        & Random & 0.7925 & 0.8718 & 0.9009  \\
\specialrule{0.12em}{1pt}{1pt}
\end{tabular}}
\label{tab:ablation}
\end{table}

\subsection{RQ3: Parameter Sensitivity Analysis}
Next, we provide more analysis and discussion on the impact of codebook size \(T\) and noise scale \(b\) in FedMM.

\subsubsection{Impact of Codebook Size \(T\)} To investigate the impact of the discrete representation capacity on model performance, we conduct a sensitivity analysis on the codebook size \(T\), varying it within the range of \{64, 128, 256, 512\}. As illustrated in Figure~\ref{fig:hyper}, we observe a consistent trend across all datasets and backbone architectures: the recommendation performance, measured by overall AUC, initially improves as \(T\) increases. Specifically, increasing the codebook size from 64 to 256 yields significant performance gains. This trajectory indicates that a larger codebook size enhances the expressiveness of the quantized representations, enabling the model to capture more fine-grained semantic clusters and subtle cross-market collaborative patterns that are otherwise lost with a smaller vocabulary. However, performance hits a saturation point and begins to degrade or plateau when \(T\) is further increased to 512. This decline suggests that an excessively large codebook space may lead to codebook collapse or overfitting, in which the model captures market-specific noise rather than generalizable features, or to sparse code utilization that hinders effective knowledge sharing. Consequently, we select \(T\)=256 as the optimal setting for our experiments, as it strikes the best balance between representation granularity and generalization robustness.

\begin{figure}[htbp] 
\centering
    \includegraphics[width=1.0\linewidth]{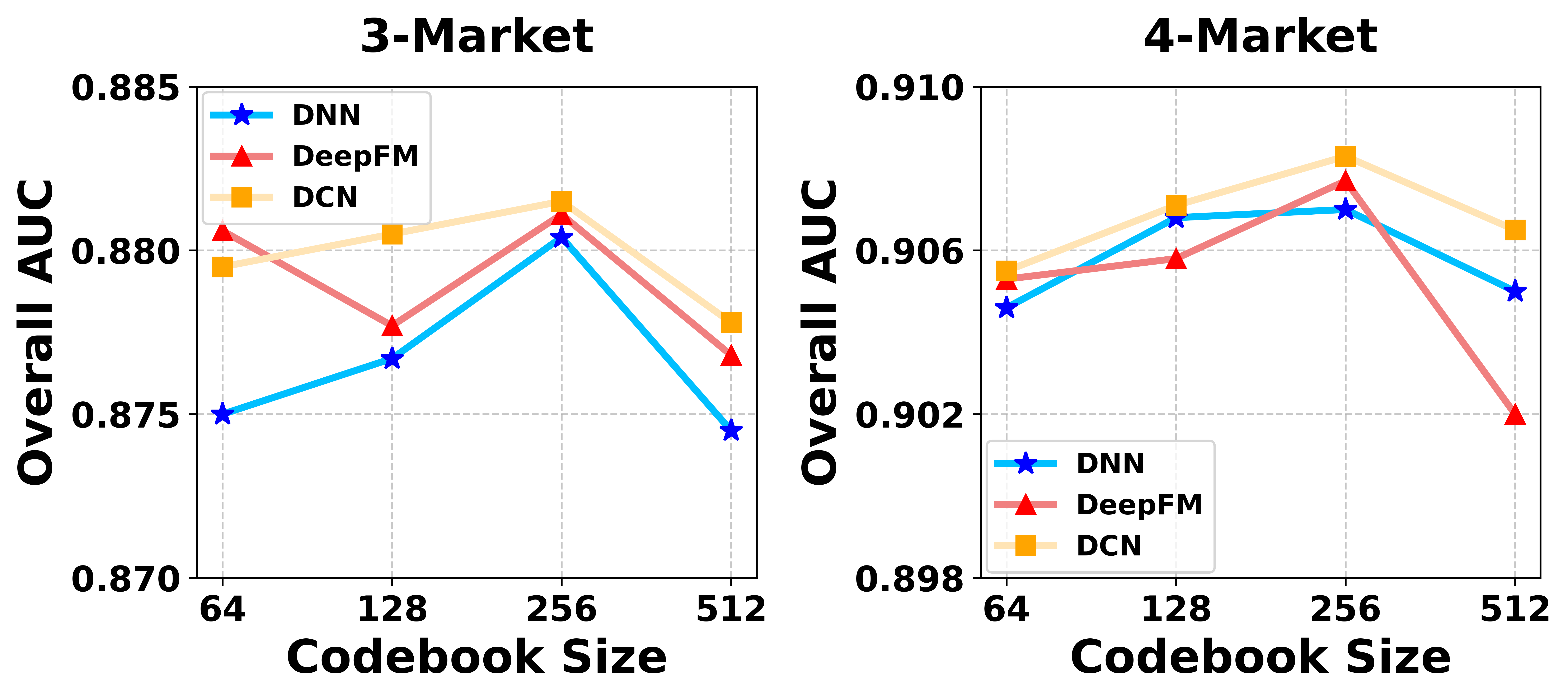}
    \Description[<Figure 4. Fully described in the text.>]{<A full description of Figure 4 can be found in Section 6.4.1.>}
\caption{Sensitivity analysis of the codebook size \(T\) on two datasets. The colored curves represent the performance trends of FedMM when integrated with different backbone architectures (blue represents DNN, red represents DeepFM, and yellow represents DCN).}
\label{fig:hyper}
\end{figure}

\subsubsection{Impact of Noise Scale \(b\)}
The noise scale \(b\) directly controls the magnitude of the noise added to the uploaded client parameters in local differential privacy (LDP). We evaluate the impact of the Laplace noise scale \(b\) on model performance. As illustrated in Figure~\ref{fig:noise}, a clear trade-off exists between privacy protection and recommendation utility. Across all models and market scenarios, increasing the noise scale b from 0.001 to 0.5 consistently degrades the Overall AUC due to the higher variance introduced by stronger privacy guarantees. Nevertheless, this performance degradation remains relatively controlled. Our framework maintains competitive accuracy even under strict privacy settings, demonstrating its effectiveness in balancing privacy and utility. 
\begin{figure}[htbp] 
\centering
    \includegraphics[width=1.0\linewidth]{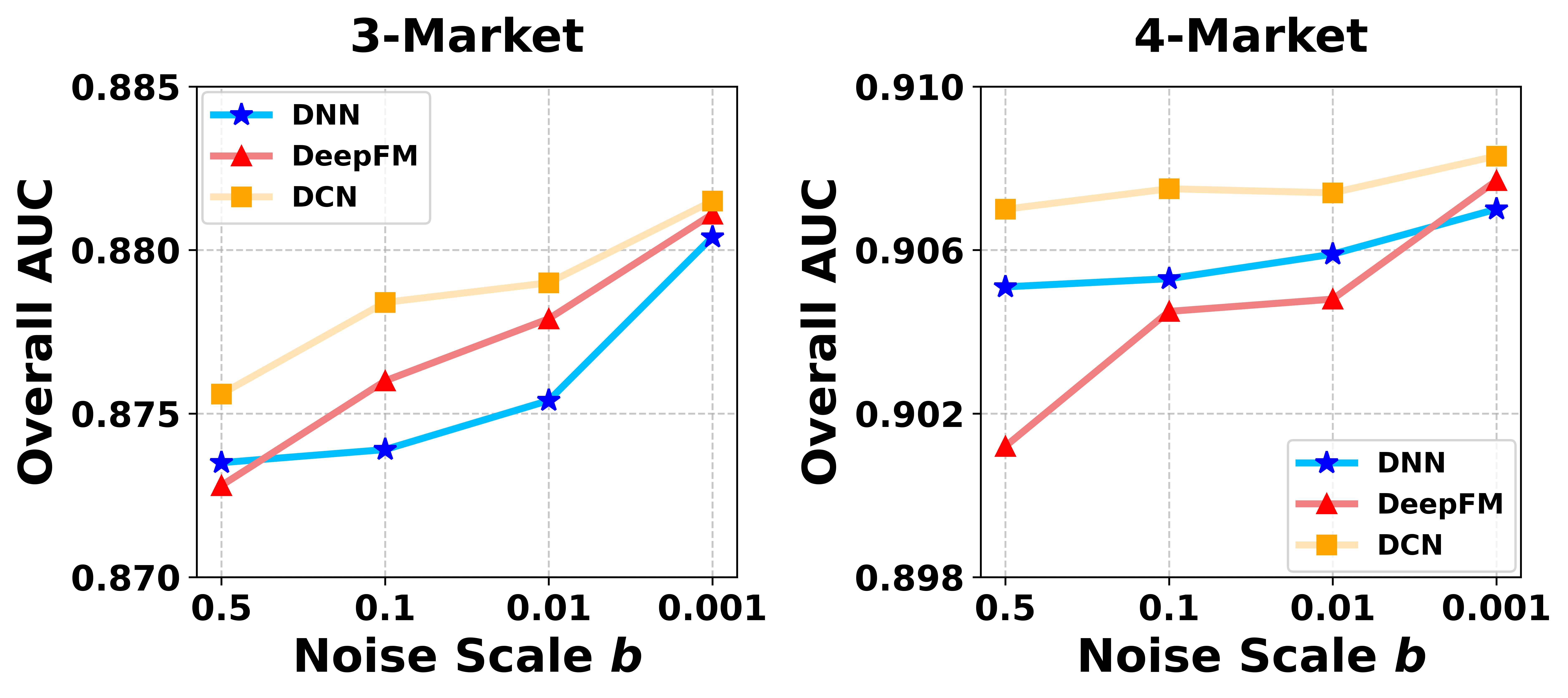}
    \Description[<Figure 5. Fully described in the text.>]{<A full description of Figure 5 can be found in Section 6.4.2.>}
\caption{Privacy-utility analysis on two datasets (relative to noise scale \(b\)). The colored curves represent the performance trends of FedMM when integrated with different backbone architectures (blue represents DNN, red represents DeepFM, and yellow represents DCN).}
\label{fig:noise}
\end{figure}

\subsection{RQ4: In-depth Analysis}
Finally, to intuitively understand how FedMM bridges diverse markets while preserving local characteristics, we visualize the learned latent spaces using t-SNE. As illustrated in Figure~\ref{fig:visual}, the user embeddings from different markets exhibit distinct clustering patterns with only partial overlaps, visually confirming the severe non-independent identically distributed (non-IID) nature and statistical heterogeneity inherent in cross-market recommendations. Crucially, the global codebook vectors are not isolated or collapsed into a single region; instead, they are predominantly concentrated in the central intersection region where diverse market manifolds converge. This spatial arrangement demonstrates that the global codebook successfully captures the universal semantic prototypes shared among diverse markets. By acting as ``connective hubs,'' these codebook vectors bridge the distributional gaps between markets, enabling the model to transfer high-level collaborative signals (the red triangles) to local clients while allowing each market to maintain its distinct manifold (the colored clusters). This confirms our hypothesis that the Federated Codebook in FedMM effectively extracts the global collaborative signal. Furthermore, the embeddings of local users across different markets exhibit a relatively dispersed pattern, whereas those within the same market show a clear clustering pattern. This further illustrates the rationality and necessity of the hierarchical codebook design.

\begin{figure}[htbp] 
\centering
    \includegraphics[width=1.0\linewidth]{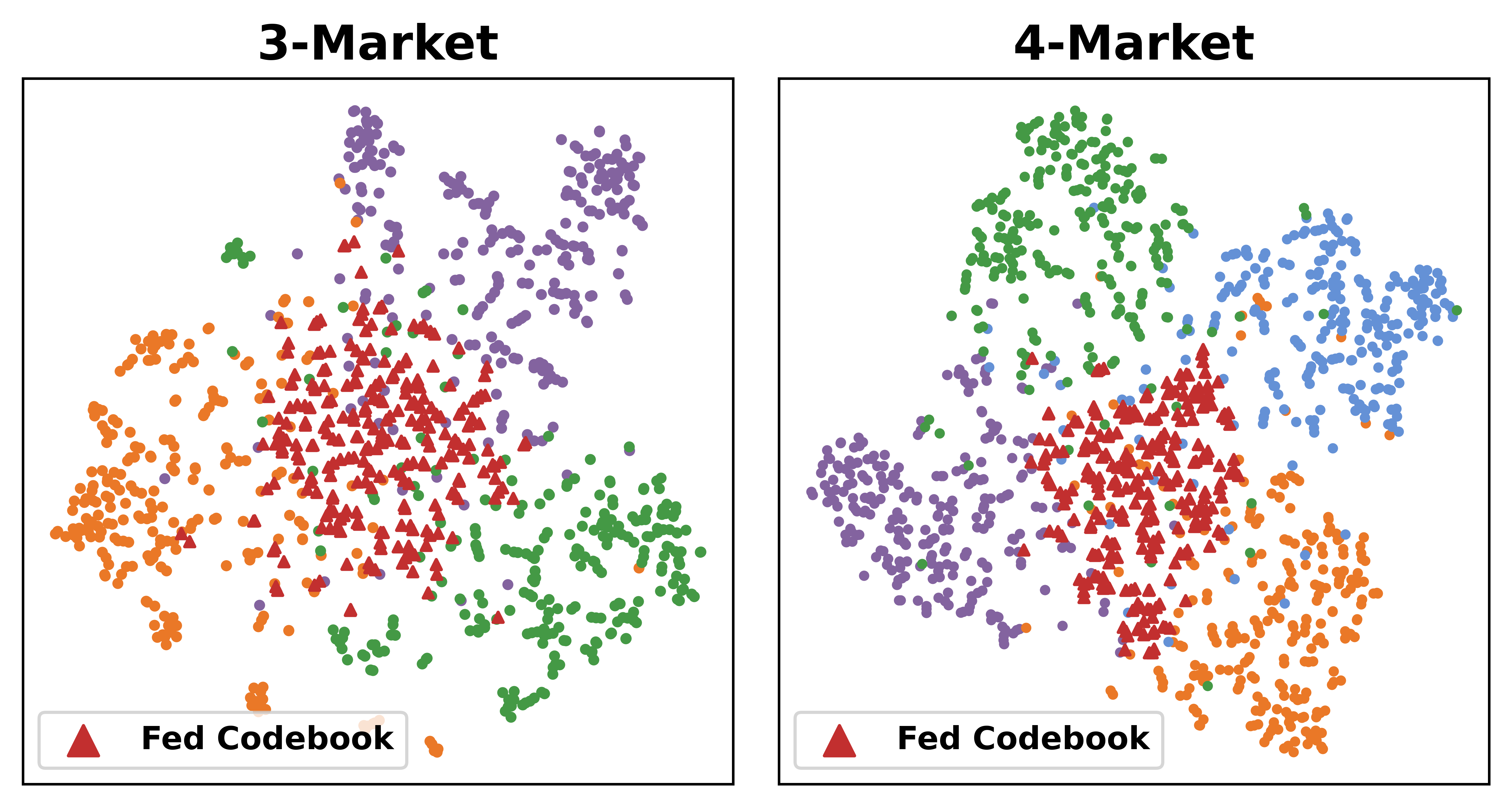}
    \Description[<Figure 6. Fully described in the text.>]{<A full description of Figure 6 can be found in Section 6.5.>}
\caption{t-SNE visualization of the latent space. The scatter points of different colors represent the local user embeddings from distinct markets, illustrating market heterogeneity. The red triangles represent the shared global codebook embeddings learned by FedMM, which are distributed across clusters to capture common collaborative patterns.}
\label{fig:visual}
\end{figure}

%% file: section/conclusion.tex
\section{Conclusions}\label{sec:conclusions}
In this paper, we propose FedMM, a novel federated framework designed for multi-market recommendation. FedMM employs a dual-codebook architecture, utilizing a Federated Codebook to extract universal collaborative prototypes shared across markets, while leveraging Local Codebooks to preserve the fine-grained characteristics of specific markets. This design enables local clients to incorporate hierarchical collaborative signals into discrete token representations, which serve as augmented features to significantly improve the CTR prediction performance.
Our codebook-based transmission strategy inherently prevents the leakage of raw data and individual user features, while substantially reducing communication overhead. Extensive experiments on three real-world datasets demonstrate that FedMM not only outperforms state-of-the-art federated baselines, but also achieves performance comparable to centralized training paradigms.

%% file: section/acknowledgement.tex
\begin{acks}
This work was supported by the National Natural Science Foundation
of China (Nos. 62302310, 62272315) and the SZTU University Research
Project (Grant No. 20251061020002).
\end{acks}